\documentclass[preprint]{elsarticle}
\usepackage{graphicx} 
\usepackage{quiver}
\usepackage{amsmath}
\usepackage{amsthm}
\usepackage{tikz}
\usepackage{xcolor}
\usetikzlibrary{calc,decorations.pathmorphing,shapes}
\usepackage{wrapfig}
\usepackage{enumitem}
\usepackage{lmodern}

\newtheorem{theorem}{Theorem}
\newtheorem{lemma}[theorem]{Lemma}  
\newtheorem{definition}{Definition}
\newtheorem{corollary}[theorem]{Corollary}

\newcounter{sarrow}
\newcommand\xrsquigarrow[1]{%
\stepcounter{sarrow}
\mathrel{\begin{tikzpicture}[baseline= {( $ (current bounding box.south) + (0,0.0ex) $ )}]
\node[inner sep=.5ex] (\thesarrow) {$\scriptstyle #1$};
\path[draw,<-,decorate,
  decoration={zigzag,amplitude=0.7pt,segment length=1.2mm,pre=lineto,pre length=4pt}] 
    (\thesarrow.south east) -- (\thesarrow.south west);
\end{tikzpicture}}%
}

\newcommand{\mynote}[3]{
    \fbox{\bfseries\sffamily\scriptsize#1}
{\small$\blacktriangleright$\textsf{\emph{\color{#3}{#2}}}$\blacktriangleleft$}}

\newcommand{\ignore}[1]{}

\newif\ifshownotes
\shownotesfalse     

\ifshownotes
\newcommand{\ps}[1]{\mynote{Pierre}{#1}{red}}
\newcommand{\pk}[1]{\mynote{Petr}{#1}{cyan}}
\newcommand{\gt}[1]{\mynote{Guillermo}{#1}{blue}}
\else
\newcommand{\ps}[1]{}
\newcommand{\pk}[1]{}
\newcommand{\gt}[1]{}
\fi

\title{Pending Conflicts Make Progress Impossible}

\author[tp]{Petr Kuznetsov}
\ead{petr.kuznetsov@telecom-paris.fr}

\author[tsp]{Pierre Sutra}
\ead{pierre.sutra@telecom-sudparis.eu}

\author[tp,tsp]{Guillermo Toyos-Marfurt \corref{cor1}}
\ead{guillermo.toyos@telecom-paris.fr}

\address[tp]{LTCI, T\'el\'ecom Paris, Institut Polytechnique de Paris, 91120 Palaiseau, France}
\address[tsp]{SAMOVAR, T\'el\'ecom SudParis, Institut Polytechnique de Paris, 91120 Palaiseau, France}

\cortext[cor1]{Corresponding author}

\begin{document}
\begin{abstract}
In this work, we study progress conditions for commutativity-aware, linearizable implementations of shared objects.
Motivated by the observation that commuting operations can be executed in parallel, we introduce \emph{conflict-obstruction-freedom}: a process is guaranteed to complete its operation if it runs for long enough without encountering step contention with conflicting (non-commuting) operations.
This condition generalizes obstruction-freedom and wait-freedom by allowing progress as long as step contention is only induced by commuting operations.
We prove that conflict-obstruction-free universal constructions are impossible to implement in the asynchronous read-write shared memory model for $n\geq3$ processes.
This result exposes a fundamental limitation of conflict-aware universal constructions: the mere invocation of conflicting operations imposes a synchronization cost.
Progress requires eventual resolution of pending conflicts.
\end{abstract}

\begin{keyword}
Distributed computing \sep Shared memory \sep Commutativity
\end{keyword}

\maketitle

\section{Introduction}

\paragraph{Context}
Synchronizing operations that access shared data is a crucial task in distributed systems.
The conventional safety property, linearizability~\cite{HW90-lin}, requires that each operation takes effect instantaneously, at some point in time between its invocation and response.
Jointly with safety, one also needs to define a progress property specifying under which circumstances an operation should return. 
The strongest of such guarantees is wait-freedom~\cite{Her91}, which demands that every operation completes in a finite number of steps of its invoking process.

We consider the standard asynchronous shared-memory model, where processes communicate via read-write primitives and may crash at any time.
%
In this model, it is well-established that wait-free linearizable implementations of many standard data structures are impossible~\cite{Her91, flp}.
Because linearizability is usually retained to preserve usability, much research has focused on exploring weaker progress guarantees. 
A prominent example is obstruction-freedom~\cite{of}:
a process that runs \emph{alone} (without encountering steps of concurrent processes) for long enough is guaranteed to complete its operation.
Although obstruction-free algorithms guarantee progress only without contention, they often behave as wait-free under typical scheduling assumptions and admit simpler implementations~\cite{of-wf}.
In particular, one can build an obstruction-free \emph{universal construction}, which is a universal recipe to create linearizable objects, using just read-write communication primitives~\cite{of,ofc}.
%

\paragraph{Motivation}
To ensure they apply to a wide variety of shared objects, progress properties are typically agnostic to their semantics. 
%
However, a key observation is that the main cost in synchronizing concurrent accesses comes from computing a total order among \emph{conflicting} (non-commuting) operations~\cite{HW90-lin}.
In contrast, commuting operations~\cite{wheil}, such as increments on a shared counter, peek operations on a list, or updates to distinct keys in a hash table, may proceed independently, in a coordination-free manner.
Some works leverage this to reduce the cost of data replication~\cite{lamport2005generalized}, boost transactional memory~\cite{boosting}, and define smart locks~\cite{FEKETE199065}.
In light of these results, a natural question is whether commuting operations may enjoy stronger progress than regular, conflicting ones.

\paragraph{Contributions}
%
%
We introduce a novel progress condition, \emph{conflict-obstruction-freedom} (COF).
COF guarantees that a process that encounters no concurrent steps of conflicting operations must be able to complete its operation.
This progress condition generalizes obstruction-freedom in a natural sense.
Moreover, for objects that only provide commuting operations, COF is identical to wait-freedom.  
%
%
For example, consider a collection of \emph{queues}, where each of them is implemented by an obstruction-free algorithm~\cite{of} and each process
can only enqueue its own identifier.
The resulting composition of queues is conflict-obstruction-free, as any two concurrent operations are in conflict if and only if they are applied to the same queue. 

Further, we investigate whether conflict-obstruction-freedom admits a universal construction in the read-write shared memory model.
%
%
%
This model captures minimal communication assumptions in a concurrent system, making it a natural starting point for studying a new progress condition.
%
%
%
%
When all operations conflict, one can build obstruction-free universal constructions~\cite{of}.
On the other hand, when all operations commute, wait-freedom is trivially achievable~\cite{pram}.
%
%
Thus, our study aims to characterize whether such a construction is possible in the intermediate cases, that is, when some operations commute and some do not. 

Our main result is negative.
We show that such a universal construction can be used to solve a weaker version of consensus in which processes must terminate when only processes with the same input remain active.
We prove that this task is impossible to solve using only reads and writes for three or more process.

This impossibility exposes a fundamental limitation in leveraging operation commutativity in concurrent systems.
The mere presence of conflicting operations incurs a synchronization cost, even if these operations later become inactive.
While commutativity eases coordination, it does not eliminate the need for synchronization: once a conflict appears, the system must eventually resolve it to make progress. 
Thus, to date, obstruction-freedom remains the strongest known progress condition that enables a universal construction in the read-write model.
\gt{Strictly speaking this claim is outdated :)}

The model and results presented in this work have been formalized in Lean~4, available at~\cite{COFLean}.
In particular, every numbered lemma, theorem, and corollary is accompanied by a corresponding machine-checked proof.
The Lean formalizations and proofs were developed with the assistance of Claude Opus~5.5. 
We provide these formalizations as supplementary material, offering an additional layer of machine-checked verification for the results of this work. 
They complement the natural-language proofs in the paper, which remain the primary presentation of the mathematical arguments.


%
%
%

\section{Preliminaries}
\label{sec:preliminaries}

\paragraph{Objects}
A sequential object is defined as a tuple $(Q,q_0,O,R,\sigma)$, where $Q$ is a set of \emph{states}, $q_0\in Q$ is the \emph{initial state}, $O$ is a set of \emph{operations}, $R$ is a set of \emph{responses}, and $\sigma:\;O\times Q \rightarrow R\times Q$ is a \emph{state transition function} specifying the \emph{response} and \textit{next state} for each operation and state.
Given a state $q$ and a sequence of operations $s\in O^*$,  we denote by $s(q)$ the final state after applying the operations in $s$ sequentially starting from $q$.
%
%
For two sequences $s$ and $s'$, $s \cdot s'$ denotes their concatenation.
Given two operations $o, o' \in O$, we say that they \emph{commute} in a state $q$ if $o \cdot o'(q) = o' \cdot o(q)$ and each operation returns the same result in both sequences.
The object defines a symmetric \emph{conflict relation}, written $\asymp~\subseteq O \times O$, where $o \asymp o'$ if there exists a state $q$ in which $o$ and $o'$ do not commute.
When this happens, we say that operations $o$ and $o^\prime$ \emph{conflict}.

\paragraph{System Model}
We consider a system with $n$ asynchronous processes, $\{p_i\}_{i \in [n]}$, that communicate through shared memory using atomic read-write registers.
Processes are deterministic and may fail by crashing, i.e., stopping prematurely and taking no further steps; an arbitrary number of processes may crash.
There is no bound on relative process speeds, and no assumptions are made about timing or scheduling fairness.

\paragraph{Algorithms and executions}
To \emph{implement} a shared object, processes follow an \emph{algorithm} consisting of deterministic automata, one per process.
When a process $p$ \emph{invokes} an operation on the object, it takes \emph{steps}, according to its automaton.
Each step consists of local computation and one shared memory event -- either a read or a write to an atomic register.
After each step, $p$ updates its local state according to its state machine and optionally returns a response on the pending high-level operation. 
In an \emph{operation instance} $\Phi =(x,o,p)$, process $p$ performs the operation $o$ on object $x$.
An \emph{execution fragment} is a (possibly infinite) sequence of steps by the processes, according to their respective state machines.
A \emph{configuration} represents a complete system state, that is, the local states of all processes together with the state of the shared memory.
%
%
An \emph{execution} is an execution fragment starting from the algorithm's initial configuration.
If the last step of an operation instance $\Phi$ in execution $\alpha$ returns a response, we say that $\Phi$ \emph{completes} in $\alpha$, otherwise $\Phi$ is \emph{pending} in $\alpha$.
We say that $\Phi$ takes a step in an execution $ \alpha $ when the process $p$ that invoked $\Phi$ takes a step and $\Phi$ has not yet returned a response in $\alpha$.
In an infinite execution, a process is \emph{correct} if it takes an infinite number of steps or it has no pending operation; otherwise, it is \emph{faulty}. 
%

\paragraph{Progress conditions}
An implementation is \emph{wait-free} (WF) if, in every execution, each 
operation instance invoked by a correct process completes in a finite number of its own steps.
%
%
%
An operation instance $\Phi$ is \emph{eventually step-contention free} in an execution $\alpha$ if either $\Phi$ completes in $\alpha$, or there exists a suffix of $\alpha$ in which only $\Phi$ takes steps (i.e., from some point on in $\alpha$, $\Phi$ runs solo).
%
An implementation is obstruction-free (OF) if, for every infinite execution $\alpha$, any operation instance $\Phi$ invoked by a correct process completes in $\alpha$ whenever $\Phi$ is eventually step-contention-free in $\alpha$~\cite{ofc}.

\section{Conflict-obstruction-freedom} \label{sec:cof}

%
%
%
%
%
We now introduce \emph{conflict-obstruction-freedom} (COF), a natural generalization of obstruction-freedom that leverages operation commutativity.
Intuitively, while obstruction-freedom guarantees termination in the absence of \emph{step contention}, COF only requires the absence of step-contention with \emph{conflicting} operations.
Thus, any operation that takes sufficiently many steps completes, provided it eventually interleaves only with steps of operations that commute with it.
%
Formally, we generalize the notion of \emph{obstruction-freedom} from~\cite{ofc} by restricting step-contention to operations that do not commute. 
%

    %
    %
    %
    %

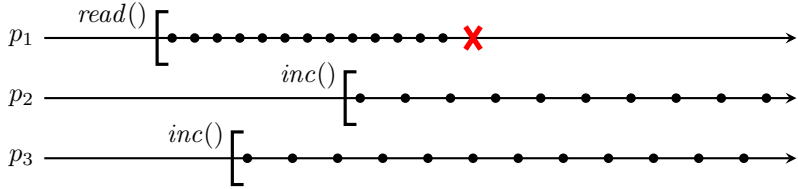
\begin{figure}
    \centering
\begin{tikzpicture}[>=stealth, thick]
  \def\yone{0}
  \def\ytwo{-0.8}
  \def\ythree{-1.6}
  \def\xstart{0}
  \def\xend{10}

  \draw[->] (\xstart,\yone) -- (\xend,\yone);
  \draw[->] (\xstart,\ytwo) -- (\xend,\ytwo);
  \draw[->] (\xstart,\ythree) -- (\xend,\ythree);

  \node[left] at (\xstart,\yone) {$p_1$};
  \node[left] at (\xstart,\ytwo) {$p_2$};
  \node[left] at (\xstart,\ythree) {$p_3$};

  \node[anchor=south west] at (0.3,\yone) {\textit{read}$()$};
  \node[anchor=south west] at (3,\ytwo) {\textit{inc}$()$};
  \node[anchor=south west] at (1.5,\ythree) {\textit{inc}$()$};

  \newcommand{\leftbracket}[4]{
    \pgfmathsetmacro{\x}{#1}
    \pgfmathsetmacro{\y}{#2}
    \pgfmathsetmacro{\h}{#3}
    \pgfmathsetmacro{\t}{#4}
    \draw[line width=1.2pt] (\x+\t,\y+\h/2) -- (\x,\y+\h/2) --  (\x,\y-\h/2) -- (\x+\t,\y-\h/2);
  }

  \leftbracket{1.5}{\yone}{0.7}{0.15}
  \leftbracket{4.0}{\ytwo}{0.7}{0.15}
  \leftbracket{2.5}{\ythree}{0.7}{0.15}

  \foreach \x in {1.7, 2, 2.3, 2.6, 2.9, 3.2, 3.5, 3.8, 4.1, 4.4, 4.7, 5.0, 5.3} {
    \fill (\x,\yone) circle (1.8pt);
  }
  \foreach \x in {4.2,4.8,5.4,6.0,6.6,7.2,7.8,8.4, 9.0, 9.6} {
    \fill (\x,\ytwo) circle (1.8pt);
  }
  \foreach \x in {2.7, 3.3, 3.9, 4.5, 5.1, 5.7, 6.3, 6.9, 7.5, 8.1, 8.7, 9.3} {
    \fill (\x,\ythree) circle (1.8pt);
  }

  \draw[red, line width=2] (5.6,0.15) -- (5.8,-0.15);
  \draw[red, line width=2] (5.6,-0.15) -- (5.8,0.15);
  
\end{tikzpicture}
\caption{
Execution of three operations on a shared counter.
Open brackets mark invocations and dots mark steps.
Process $p_1$ crashes after some steps, while $p_2$ and $p_3$ continue concurrently.
A read conflicts with an increment.
Two increments commute.
Conflict-obstruction-freedom guarantees completion of the two increments, while obstruction-freedom does not.
%
%
%
%
}
\label{fig:cof_example}
\vspace{-12pt}
\end{figure}

\begin{definition}
An operation instance $\Phi\!=\!(x,o,p)$ is eventually conflict-step-contention free in an execution $\alpha$ if either $\Phi$ completes in $\alpha$, or there exists a suffix of $\alpha$ in which no other operation $\Phi'=(x,o',q)$ takes steps such that $o\!\asymp\!o'$.

\end{definition}

\begin{definition}
An implementation is conflict-obstruction-free if, for every infinite execution $\alpha$ and every operation instance $\Phi$ invoked by a correct process such that $\Phi$ is eventually conflict-step-contention free in $\alpha$, $\Phi$ completes in $\alpha$.

\end{definition}

To grasp the new progress condition, consider the execution in Figure~\ref{fig:cof_example}, where three processes access a shared counter.
%
%
The object supports two types of operations: $read()$, which returns the current value of the counter, and
$inc()$, which increments the counter by one.
%
%
In the figure, process $p_1$ invokes $o_1=read()$, while processes $p_2$ and $p_3$ invoke $o_2=inc()$ and $o_3=inc()$, respectively.
Reading a counter conflicts with incrementing it ($o_1\!\asymp\! o_2$ and $o_1\!\asymp\!o_3$), but two increment operations --- and likewise two read operations --- commute with each other ($o_2 \not\asymp o_3$).
After some steps of $p_1$, $p_1$ crashes 
while $p_2$ and $p_3$ continue concurrently.
In Figure~\ref{fig:cof_example}, no operation is eventually step-contention free.
As a consequence, OF would not guarantee that the operations of the two correct processes complete.
On the contrary, COF would ensure progress because $o_2$ and $o_3$ are in fact eventually conflict-step-contention free.

\ignore{
Another simple example is provided by implementations of collections of objects.
Consider a collection of objects such that operations on the same object conflict, whereas operations on distinct objects commute.
Using an obstruction-free implementation for each object yields a conflict-obstruction-free implementation of the collection.
%
}

\begin{lemma}
    Let $\mathcal{I}$ be an implementation of a shared object $(Q,q_0,O,R,\sigma)$.
    If the induced conflict relation satisfies $\asymp = O\times O$ (i.e., every pair of operations conflicts), then $\mathcal{I}$ is conflict-obstruction-free iff $\mathcal{I}$ is obstruction-free.
\end{lemma}
\begin{proof}
An operation instance $\Phi$ is eventually conflict-step-contention free if either (1) $\Phi$ completes, or (2) there is a suffix of the execution in which no other
operation conflicting with $\Phi$ takes steps.
If $\asymp=\!O\!\times\!O$, any operation conflicts with $\Phi$, so condition (2) holds iff no other operation takes steps in the suffix.
Hence, in this case, eventual conflict-step-contention-freedom coincides with eventual step-contention-freedom.
Since COF requires termination of every eventually conflict-step-contention free operation, and OF requires termination of every eventually step-contention free operation, the conditions coincide.
%
\end{proof}

\begin{lemma}
    Let $\mathcal{I}$ be an implementation of a shared object $(Q,q_0,O,R,\sigma)$.
    If the induced conflict relation is empty, $\asymp = \emptyset$ (i.e., all operations commute), then $\mathcal{I}$ is conflict-obstruction-free iff $\mathcal{I}$ is wait-free.
\end{lemma}
\begin{proof}
If $\asymp=\emptyset$, no two operations conflict. 
Hence, in every execution, every operation instance is eventually conflict-step-contention free.
By definition of COF, every such operation invoked by a correct process must complete, which implies wait-freedom.
Conversely, any wait-free implementation guarantees termination of every operation in every execution, satisfying COF.
\end{proof}

\section{Impossibility of Conflict-obstruction-freedom} \label{sec:impossible}

Our main result establishes that no algorithm can implement a conflict-obstruction-free universal construction in the read-write shared memory model.
The proof proceeds by a reduction to the (binary) consensus problem requiring COF termination.
Then, we show that this problem cannot be solved in the read-write shared memory model.
%
%
%

%
%



\subsection{Conflict-obstruction-free Consensus} \label{sec:quasi_cons}

In the consensus problem~\cite{flp},  processes \emph{propose} values by invoking a single-shot operation \textit{propose($v$)} and \emph{decide} the values returned by these invocations.
A consensus algorithm satisfies the following properties:

\begin{itemize}[itemsep=0pt]
    \item \textbf{Validity:} If a process decides a value $v$, then some process proposed $v$;
    \item \textbf{Agreement:} No two processes decide differently; 
    \item \textbf{WF termination:} Every correct process decides in finitely many steps.
\end{itemize}

Conflict-obstruction-free consensus retains Validity and Agreement, but replaces wait-free termination with the following weaker liveness condition:

\begin{itemize}[itemsep=0pt]
\item \textbf{COF termination:} If eventually only processes that proposed the same value take steps, then every correct process decides in finitely many steps.
\end{itemize}


We now model consensus as a deterministic sequential object, restricting attention to the binary case. 
Its state space is $Q\!=\!\{0,1,\bot\}$ and initial state $q_0\!=\!\bot$, where $0$ and $1$ represent the states where values $0$ and $1$ have been decided, respectively, and $\bot$ represents the initial undecided state.
The set of operations is $O\!=\!\{\textit{propose}(v)\}_{v\in \{0,1\}}$, each returning a value in $\{0,1\}$.
The transition function $\sigma$ sets the state to the argument of \emph{the first} \textit{propose} operation and returns that value; subsequent operations return the decision value associated with the state without modifying it.
Hence, any pair of \textit{propose} operations \emph{commute} in states $0$ and $1$, and two operations $\textit{propose}(v)$ and $\textit{propose}(v')$ commute in state $\bot$ iff $v\!=\!v'$.
Thus, two proposals conflict if and only if they propose different values: $\textit{propose}(0) \asymp \textit{propose}(1)$.
%
Thus, a COF consensus implementation satisfies COF termination: every correct process completes its invocation in executions where, after some point, only processes with the same input value take steps.

Conflict-obstruction-free consensus is implementable in the read–write model for two processes. 
The consensus algorithm in~\cite{ofc} satisfies COF termination in this case, as it guarantees termination when both processes have the same input.
As we show below, this is no longer the case for more than two processes.

\subsection{Proof of the Impossibility Result}

We show that COF consensus cannot be implemented in the read-write shared memory model for three or more processes.
To establish this result, we consider a minimal model instance with three processes $\{p_0, p_1, p_2\}$, where $p_0$ proposes $0$ and both $p_1$ and $p_2$ propose $1$.
The proof adapts the classical FLP-style bivalency argument~\cite{flp,LA87}, using an appropriate definition of valency that accounts for COF termination.





%
%

%
The next step of each process is fully determined by its local state.
%
%
We denote a step $e = (p, m, a)$, where $p$ is the process performing the step, $m$ is the memory location accessed, and $a \in \{\text{R}, \text{W}\}$ indicates whether $p$ is reading or writing register $m$.
%
If $e$ is the next step of $p$ in some configuration $C$, we say that $e$ is \emph{enabled in $C$}.
%

An \emph{initial configuration}, specifying the state of each process and the shared memory, is characterized by the vector assigning inputs to each process.
%
In our consensus instance, we consider the initial input vector to be $I\!=\!(0,1,1)$.
%
We write $C\!\xrightarrow{e}\!C'$ or $C \xrightarrow{} C'$ if there exists an enabled step $e\!=\!(p,m,a)$ in $C$ such that process $p$ takes step $e$, as a result, the configuration of the system changes from $C$ to $C'$. 
Note that the only difference between $C$ and $C'$ is the local state of process $p$, its next enabled step, and possibly the state of the shared register $m$. 
%
We write $C\! \xrsquigarrow{\ }\!C'$ or $C \!\xrsquigarrow{\pi}\!C'$ if $C\!=\!C'$ or there exists a sequence of steps $\pi\!=\!e_1, \ldots, e_{k \geq 1}$ such that $C\!\xrightarrow{e_1} \!\dots\! \xrightarrow{e_{k}}\!C'$ and say that $C'$ is \textit{reachable} from $C$, or $C'$ is a \emph{descendant} of $C$.
For a set $P \subseteq \{p_i\}_{i \in[n]}$, we say that an execution fragment $\pi$ is \emph{$P$-only} if $\pi$ only contains steps of processes in $P$.
%

We define the valency of a configuration according to the \emph{decisions} processes $p_1$ and $p_2$ take in a subset of reachable configurations.
A configuration $C$ is $v$-\emph{deciding} if some process has decided $v$ in $C$.
We say that a configuration $C$ is \emph{$v$-valent} if 
there is no $\{p_1,p_2\}$-only sequence $\pi$, such that $C \xrsquigarrow{\pi} C'$ and $C'$ is $(1-v)$-deciding.
We say that a configuration is \emph{univalent} if it is $0$-valent or $1$-valent. 
%
Otherwise, we say that the configuration is \emph{bivalent}.  
Clearly, if a configuration $C$ is $v$-valent and it has a descendant $C'$, $C \xrsquigarrow{\pi} C'$, where $\pi$ is $\{p_1,p_2\}$-only, then $C'$ is also $v$-valent.

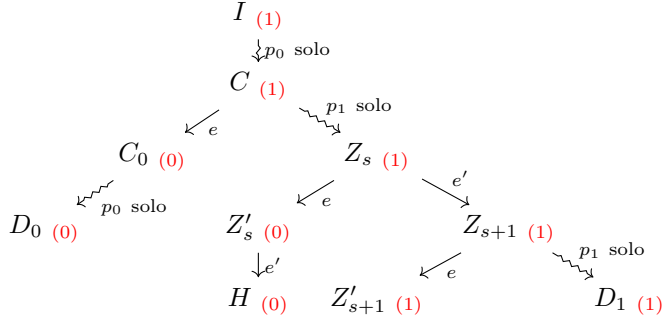
\begin{figure}
\vspace{-3pt}
\centering
\[\begin{tikzcd}[column sep=tiny,row sep=small]
	&& {I  \ _{\color{red}(1)}} \\
	&& {C  \ _{\color{red}(1)}} \\
	& {C_0  \ _{\color{red}(0)}} && {Z_s  \ _{\color{red}(1)}} \\
	{D_0\ _{\color{red}(0)}} && {Z'_s  \ _{\color{red}(0)}} && {Z_{s+1} \ _{\color{red}(1)}} \\
	&& H\ _{\color{red}(0)} & {Z_{s+1}'  \ _{\color{red}(1)}} && {D_1\ _{\color{red}(1)}}
	\arrow["{p_0 \ \text{solo}}", squiggly, from=1-3, to=2-3]
	\arrow["e", from=2-3, to=3-2]
	\arrow["{p_1 \ \text{solo}}", squiggly, from=2-3, to=3-4]
	\arrow["{p_0 \ \text{solo}}", squiggly, from=3-2, to=4-1]
	\arrow["e", from=3-4, to=4-3]
	\arrow["{e'}", from=3-4, to=4-5]
	\arrow["{e'}", from=4-3, to=5-3]
	\arrow["e", from=4-5, to=5-4]
	\arrow["{p_1 \ \text{solo}}", squiggly, from=4-5, to=5-6]
\end{tikzcd}\]

\vspace{-18pt}
\caption{Valency diagram showcasing the contradiction in Lemma~\ref{lemma:init_tobivalent}.
}
\label{fig:flp_diag}
\vspace{-12pt}
\end{figure}

\begin{lemma} \label{lemma:init_config_1valent}
    The initial configuration $I$ with input vector $(0,1,1)$ is $1$-valent.
\end{lemma}
\begin{proof}
    Consider any execution $\pi$ from $I$ in which only $p_1$ and $p_2$ take steps.
    This execution is indistinguishable to $p_1$ and $p_2$ from an execution starting from the input vector $(1,1,1)$;  therefore, by Validity, no process decides $0$ in $\pi$.
    %
\end{proof}
    


\begin{lemma} \label{lemma:init_tobivalent}
    There exists a bivalent configuration reachable from $I$. 
\end{lemma}
\begin{proof}
    Suppose, by contradiction, that every configuration reachable from $I$ is univalent.
    %
    Figure~\ref{fig:flp_diag} illustrates the key configurations and critical steps considered below.
    %
    %
    Suppose an execution $\sigma$ where $p_0$ runs solo from $I$. 
    As $p_1$ and $p_2$ take no steps, $p_0$ decides $0$ in $\sigma$ by COF Termination and Validity.
    %
    Let $D_0$ be the configuration reached in $\sigma$ where $p_0$ decides $0$.
    By Agreement (and linearizability), $D_0$ is $0$-valent.
    %
    
    By our assumption, all configurations reachable from $I$ are univalent (in particular, all configurations in $\sigma$, from $I$ to $D_0$).
    By Lemma~\ref{lemma:init_config_1valent}, $I$ is $1$-valent.
    Hence, there exist configurations $C$, $C_0$ and a \emph{critical} step $e=(p_0,m,a)$ such that $I \xrsquigarrow{\ } C \xrightarrow{e} C_0$, where $C$ is $1$-valent and $C_0$ is $0$-valent.
    Now consider an execution $\pi$ where $p_1$ runs solo from $C$ until it decides.
    By Termination and the fact that $C$ is $1$-valent, $\pi$ ends at a configuration $D_1$ where $p_1$ decides $1$.
    %
    
    Let $C=Z_0 \ldots Z_{k \geq 1}=D_1$ be the sequence of configurations obtained by applying the steps in $\pi$.
    Since $p_0$ takes no step in $\pi$, $e$ is enabled in each $Z_{i \in [0,k]}$.
    Let $Z'_i : Z_i \xrightarrow{e}Z'_i$ be the
    resulting configuration from applying step $e$ from $Z_i$.
    %
    %
    Then $Z'_0=C_0$ is $0$-valent and $Z'_k$ is $1$-valent by Agreement.
    %
    Hence, there exists $s \in [0,k\!-\!1]$ such that $Z_s'$ is $0$-valent, and $Z_s$, $Z_{s+1}$ and $Z'_{s+1}$ are $1$-valent.

    Let $e'=(p_1,m',a')$ be the step such that $Z_s \xrightarrow{\smash{e'}} Z_{s+1}$, and $H$ the configuration such that $Z'_s \xrightarrow{\smash{e'}} H$.
    %
    %
    We observe that step $e=(p_0,m,a)$ must be a write. 
    Otherwise, $H$ and $Z_{s+1}$ are indistinguishable to $p_2$, however, $H$ is $0$-valent and $Z_{s+1}$ is $1$-valent.
    Similarly, $e'$ must also be a write; otherwise $Z'_s$ and $Z'_{s+1}$ are indistinguishable to $p_2$, yet $p_2$ would decide differently from $Z'_s$ and $Z'_{s+1}$.
    %
    %
    
    Now there are two cases to consider, both leading to a contradiction.
    %
     If $m\!\neq\!m'$, then $e'$ and $e$ commute, leading to $H=Z'_{s+1}$, and contradicting that $H$ is $0$-valent while $Z'_{s+1}$ is $1$-valent.
    %
    %
    If instead $m\!=\!m'$, then $p_2$ cannot distinguish $Z'_{s+1}$ from $Z'_s$, yet $Z'_s$ is $0$-valent while $Z_{s+1}'$ is $1$-valent; a contradiction.
\end{proof}

%


\begin{lemma} \label{lemma:bivalent_to_bivalent}
    Any bivalent configuration $C\!$ with $\!I\! \xrsquigarrow{\ }\! C\!$ has an enabled step $e$ by a process in $\{p_1,p_2\}$ such that the resulting configuration $\!C'\!$ with $\!C\! \xrightarrow{e} \!C'\!$ is bivalent.

\end{lemma}

\begin{proof}
Let $e_1$ and $e_2$ be the steps by, resp., $p_1$ and $p_2$ that are enabled in $C$.
Let $C_1$ and $C_2$ be the configurations obtained by applying, resp., $e_1$ and $e_2$ to $C$.
By contradiction, assume that both $C_1$ and $C_2$ are univalent.
%
%
%
Moreover, $C_1$ and $C_2$ must have different valencies, otherwise $C$ would not be bivalent.
Without loss of generality, assume that $C_1$ is $0$-valent and $C_2$ is $1$-valent.
Let $C_1'$ and $C_2'$ be the resulting configurations of applying steps $e_2$ and $e_1$ from $C_1$ and $C_2$ respectively ($C_1 \xrightarrow{e_2} C_1'$ and $C_2 \xrightarrow{e_1} C_2'$).
If both $e_1$ and $e_2$ are read steps, configurations $C_1'$ and $C_2'$ are indistinguishable to any process and have opposite valencies---a contradiction since $p_1$ has to decide differently in any $\{p_1\}$-only execution right after $C_1'$ and $C_2'$.
The same argument applies if $e_1$ and $e_2$ access different registers.

Thus, the two steps access the same register and (at least) one of them is a write. 
Without loss of generality, suppose that $e_1$ is a write step. 
%
%
%
By termination, if we schedule $\{p_1\}$-only executions right after $C_1$ and $C_2'$, it should eventually decide.
Let $D_1$ and $D_2$ be the corresponding deciding configurations.
Process $p_1$ cannot distinguish $C_1$ from $C_2'$: if $e_2$ is a write step, it writes to the same location as step $e_1$ and occurs before it ($C\!\xrightarrow{e_2}\!C_2\!\xrightarrow{e_1}\!C_2'$), and if $e_2$ is a read step, it affects only the local state of  $p_2$.
Hence, $p_1$ cannot distinguish between $D_1$ and $D_2$, and thus should decide the same value in both configurations---a contradiction with the fact that $C_1$ (and, thus, $D_1$) is $0$-valent and $C_2$ (and, thus, $D_2$) is $1$-valent. 
\end{proof}

\begin{theorem} \label{th:no_quasi_cons}
    No algorithm implements conflict-obstruction-free consensus in the read-write shared memory model of three or more processes.
\end{theorem}
\begin{proof}
    By contradiction, suppose that there exists an algorithm $\mathcal{A}$ implementing COF consensus for three processes $\{p_0,p_1,p_2\}$. 
    Let us consider the initial configuration $I$ with the input vector $(0,1,1)$, i.e., $p_0$ proposes $0$ and $p_1$ and $p_2$ propose $1$.  
    By Lemma~\ref{lemma:init_tobivalent}, $\mathcal{A}$ has a bivalent configuration $C$ reachable from $I$.
    Lemma~\ref{lemma:bivalent_to_bivalent} shows that from any bivalent configuration we can take a step by a process in $\{p_1,p_2\}$ to reach another bivalent configuration. 
    By repeating this argument, we construct an infinite execution starting from $I$ in which, from some point on, only $\{p_1,p_2\}$ take steps.
    Hence, at least one process $p_i \in \{p_1,p_2\}$ takes infinitely many steps and is therefore correct.
    As the execution goes through bivalent configurations only, no process can decide in it.
    However, because $\{p_1,p_2\}$ have the same input, $p_i$'s operation is eventually conflict-step-contention free and is therefore required to complete ---a contradiction. 
    Hence no such algorithm exists.
\end{proof}

A COF universal construction is an algorithm that takes as input a sequential object specification and produces a linearizable shared-memory implementation satisfying COF termination --- in particular it implements COF consensus.

\begin{corollary} \label{cor:no_cof_univ}
    No universal construction algorithm is conflict-obstruction-free in the read-write shared memory model of three or more processes.
\end{corollary}

\section*{Acknowledgments} This work was supported by the CHIST-ERA European initiative (grant CHIST-ERA-22-SPiDDS-05,
REDONDA project), and by Agence Nationale de la Recherche (grant ANR-23-CHR4-0009).

\bibliographystyle{elsarticle-num}
\bibliography{references}

\end{document}